\documentclass{article}
\usepackage{geometry}
\newgeometry{textwidth=14cm}
\usepackage{amsmath}
\usepackage{amssymb}
\usepackage{amsthm}
\usepackage{algorithm}
\usepackage[noend]{algorithmic}
\usepackage{multirow}
\usepackage[table]{xcolor}
\usepackage{graphicx}
\usepackage{flushend}
\newtheorem{theorem}{Theorem}

\newtheorem{lemma}{Lemma}
\newtheorem{Definition}{Definition}
\newtheorem{proposition}{Proposition}
\usepackage{url}
\usepackage{epsfig}

\usepackage[plainpages=false,pdfpagelabels]{hyperref}

\begin{document}
\title{Aircraft Landing Problem:\\Efficient Algorithm for a Given Landing Sequence}

\author{Abhishek Awasthi$^*$, Oliver Kramer$^\dag$ and J\"org L\"assig$^*$ \\ \\
$^*$Department of Computer Science\\
University of Applied Sciences Zittau/G\"orlitz \\ G\"orlitz, Germany\\
\{abhishek.awasthi, joerg.laessig\}@hszg.de
\and 
$^\dag$Department of Computing Science\\
Carl von Ossietzky University of Oldenburg\\
Oldenburg, Germany \\
oliver.kramer@uni-oldenburg.de
}
\maketitle
\begin{abstract}
In this paper, we investigate a special case of the static aircraft landing problem (ALP) with the objective to optimize landing sequences and landing times for a set of air planes. The problem is to land the planes on one or multiple runways within a time window as close as possible to the preferable target landing time, maintaining a safety distance constraint. The objective of this well-known NP-hard optimization problem is to minimize the sum of the total penalty incurred by all the aircraft for arriving earlier or later than their preferred landing times. For a problem variant that optimizes a given feasible landing sequence for the single runway case, we present an exact polynomial algorithm and prove the run-time complexity to lie in $O(N^3)$, where $N$ is the number of aircraft. The proposed algorithm returns the optimal solution for the ALP for a given feasible landing sequence on a single runway for a common practical case of the ALP described in the paper. Furthermore, we propose a strategy for the ALP with multiple runways and present our results for all the benchmark instances with single and multiple runways, while comparing them to previous results in the literature.
\end{abstract}

\section{Introduction}
Prior to landing, an aircraft must go through an approaching stage directed by the air traffic control (ATC) tower. The ATC gives instructions to the aircraft regarding to the runway, speed and altitude of the aircraft in order to align it with the allocated runway and maintain the safety distance with its preceding aircraft. During peak hours, controllers must handle safely and effectively landings of a continuous flow of aircraft entering the radar range onto the assigned runway(s). The capacity of runways is highly constrained and this makes the scheduling of landings a difficult task to perform effectively. Increasing the capacity of an airport by building new runways is an expensive and difficult affair. Hence, the air traffic controllers face the problem of allocating a landing sequence and landing times to the aircraft in the radar range. Additionally, in case of airports with multiple runways, they have to take a decision on the runway allotment too, {\em i.e.\ \/}which aircraft lands on which runway. These decisions are made with the availability of certain information about the aircraft in the radar range~\cite{krishna,earnst,pinol}. A target landing time is defined as the time at which an aircraft can land if it flies straight to the runway at its cruise speed (most economical). This target landing time is bounded by an earliest landing time and a latest landing time known as the time window. The earliest landing time is determined as the time at which an aircraft can land if it flies straight to the runway at its fastest speed with no holding, and the latest landing time is determined as the time at which an aircraft can land after it is held for its maximal allowable time before landing. All the aircraft have to land within their time window and there are asymmetric penalties associated with each aircraft for landing earlier or later than its target landing time. Besides, there is another constraint of the safety distance that has to be maintained by any aircraft with its preceding aircraft. This separation is necessary as every aircraft creates a wake vortices at its rear and can cause a serious aerodynamic instability to a closely following aircraft. There are several types of planes which land on a runway hence the safety distance between any two aircraft depends on their types. This safety distance between any two aircraft can be easily converted to a safety time by considering the required separation and their relative speeds. If several runways are available for landing, the application of this constraint for aircraft landing on different runways usually depends upon the relative positions of the runways~\cite{krishna,earnst,pinol}. A formal definition of the ALP is given in Section~\ref{probform}.

The objective of the ALP is to minimize the total penalty incurred due to the deviation of the scheduled landing times of all the aircraft with their target landing times. Hence, the air traffic controllers not only have to find suitable landing times for all the aircraft but also a landing sequence so as to reduce the total penalty. This work considers the static case of the aircraft landing problem where the set of aircraft that are waiting to land is already known. For a special but practically very common case of the safety constraint, we present a polynomially bound exact algorithm for optimizing any given feasible landing sequence for the single runway case and an effective strategy for the multiple runway case. In the later part of the paper we present our results for all the benchmark instances provided by Beasley~\cite{beasley1} and compare the results with previous work on this problem.

\section{Related Work}
The aircraft landing problem described in this paper was first introduced and studied by Beasley in the mid-nineties~\cite{firstalp}. Since then, it has been studied by several researchers using different metaheuristics, hybrid metaheuristics, linear programming, variants of exact branch and bound algorithms~\emph{etc.,} for both the static and dynamic cases of the problem. In $1995$, Beasley~\emph{et al.} presented a mixed-integer zero-one formulation of the problem for the single runway case and later extended it to the multiple runway case~\cite{firstalp}. The ALP was studied for up to $50$ aircraft with multiple runways using linear programming based tree search and an effective heuristic algorithm for the problem. Again in $1995$, Abela~\emph{et al.}~\cite{abela} proposed a genetic algorithm and a branch and bound algorithm to solve the problem of scheduling aircraft landings. Ernst~\emph{et al.} presented a simplex algorithm which evaluated the landing times based on some partial ordering information. This method was used in a problem space search heuristic as well as a branch-and-bound method for both, the single and multiple runway case, for again up to $50$ aircraft~\cite{earnst}. Beasley~\emph{et al.} adopted similar methodologies and presented extended results~\cite{krishna}. In $1998$, Ciesielski~\emph{et al.} developed a real time algorithm for the aircraft landings using a genetic algorithm and performed experiments on landing data for the Sydney airport on the busiest day of the year~\cite{paul}. In $2001$, Beasley~\emph{et al.} developed a population heuristic and implemented it on actual operational data related to aircraft landings at the London Heathrow airport~\cite{beasley}. The dynamic case of the ALP was studied again by Beasley~\emph{et al.} by expressing it as a displacement problem and using heuristics and linear programming~\cite{beasleyDis}. In $2006$, Pinol and Beasley presented two heuristic techniques, Scatter Search and the Bionomic Algorithm and published results for the available test problems involving up to $500$ aircraft and $5$ runways~\cite{pinol}. The dynamic case of the problem for the single-runway case was again studied by Moser~\emph{et al.} in $2007$~\cite{moser}. They used extremal optimization along with a deterministic algorithm to optimize a landing sequence. In $2008$ Tang~\emph{et al.} implemented a multi-objective evolutionary approach to simultaneously minimize the total scheduled time of arrival and the total cost incurred~\cite{tang}. In $2009$, Bencheikh~\emph{et al.} approached the ALP using hybrid methods combining genetic algorithms and ant colony optimization by formulating the problem as a job shop scheduling problem~\cite{bencheikh2}. The same authors presented an ant colony algorithm along with a new heuristic to adjust the landing times of the aircraft in a given landing sequence in order to reduce the total penalty cost, in $2011$~\cite{bencheikh}. In $2012$, a hybrid meta-heuristic algorithm was suggested using simulated annealing with variable neighbourhood search and variable neighbourhood descent~\cite{salehipour}.

\section{Problem Formulation}\label{probform}

In this section we give the mathematical formulation of the static aircraft landing problem based on~\cite{pinol}. We also define some new parameters which are later used in the presented algorithm in the next sections.

\noindent
Let,\\
$N$ = the number of aircraft, \\
$E_{i}$ = the earliest landing time for aircraft $i$, $i = 1,2,\dots,N$,  \\
$L_{i}$ = the latest landing time for aircraft $i$, $i = 1,2,\dots,N$,  \\
$T_{i}$ = the target landing time for aircraft $i$, $i = 1,2,\dots,N$,  \\
$ST_{i}$ = the scheduled landing time for aircraft $i$, \\
$S_{i,j}$ = the required separation time between planes $i$ and $j$, where plane $i$ lands before plane $j$ on the same runway, $i \neq j$, \\
$s_{i,j}$ = the required separation time between planes $i$ and $j$, where plane $i$ lands before plane $j$ on different runways, $i \neq j$, \\
$g_{i}$ = the penalty cost per time unit associated with plane $i$ for landing before $T_{i}$, \\ 
$h_{i}$ = the penalty cost per time unit associated with plane $i$ for landing after $T_{i}$, \\ 
$\alpha_{i}$ = earliness (time) of plane $i$ from $T_{i}$, $\alpha_{i}=\max\{0,T_{i}-ST_{i}\}$, $i = 1,2,\dots,N$, \\
$\beta_{i}$ = tardiness (time) of plane $i$ from $T_{i}$, $\beta{i}=\max\{0,ST_{i}-T_{i}\}$, $i = 1,2,\dots,N$\;.\\ \\

The total penalty corresponding to any aircraft $i$ is then expressed as $\alpha_{i}g_{i}+\beta_{i}h_{i}$. If aircraft $i$ lands at its target landing time then both $\alpha_{i}$ and $\beta_{i}$ are equal to zero and the cost incurred by its landing is equal to zero. However, if aircraft $i$ does not land at $T_{i}$, either $\alpha_{i}$ or $\beta_{i}$ is non-zero and there is a strictly positive cost incurred. The objective function of the problem can now be defined as
\begin{equation}\label{ob}
\min \sum\limits_{i=1}^{N} (\alpha_{i}g_{i}+\beta_{i}h_{i})\;.
\end{equation}

\section{The Exact Algorithm}
In this section we present our exact polynomial algorithm for the aircraft landing problem with a special case of the safety constraint for the single runway case. Before stating the algorithm we first present some new parameters, definitions and lemmas which are useful for its explanation. We first define $D_{i}$ as the algebraic deviation of the scheduled landing time of plane $i$ from its target landing time, $D_{i} = ST_{i}-T_{i}$, $i=1,2,\dots,N$. We also define $ES_{i}$ as the minimum of extra separation times maintained between plane $i$ and all its preceding planes, and the deviation from its earliest landing time, for $i>1$. For $i=1$ we define $ES_{i}$ as the deviation of its scheduled landing time with its earliest landing time, as there are no planes landing before the first plane. Mathematically, $ES_{i}$ can be written as 
\begin{equation}\label{est}
ES_{i}=
\begin{cases}
 ST_{i}-E_{i}, & \mbox{if }i=1, \\
ST_{i}- \max\{SP_{i},E_{i}\}, &  \mbox{if }2 \leq i \leq  N,
\end{cases}
\end{equation}
where,
\begin{equation}\label{sp}
SP_{i}= \max_{j=1,2, \dots, i-1} (ST_{j} + S_{j,i}) \;.
\end{equation}

Here, $SP_{i}$ is the time at which an aircraft $i$ can land maintaining the safety constraint with all its preceding planes.
Let $P$ be a given landing sequence of the planes where the $i$th plane in this sequence is denoted by ${i}$, $i=1,2,\dots, N$. Note that without loss of any generality we can assume this, since the planes can be ranked in any sequence as per their order of landing.

\begin{lemma}\label{1lemma}
If the initial assignment of the landing times of all the aircraft in any landing sequence for a single runway is done according to $ST_{i}$, where 
\begin{equation}\label{lemma1}
ST_{i} = 
\begin{cases}
L_{i}, & \mbox{if } i=N \\ 
\min\{PS_{i},L_{i}\} & \mbox{if } 1 \leq i \leq N-1,
\end{cases}
\end{equation}
where,
\begin{equation}\label{ps}
PS_{i}= \min_{j=i+1,i+2, \dots, N} (ST_{j} - S_{i,j}),
\end{equation}
then the optimal solution can be obtained only by reducing the landing times of the aircraft while respecting the constraints or leaving the landing times unchanged.
\end{lemma}

\begin{proof}
Equation~\eqref{lemma1} schedules the landing times of the aircraft in the reverse order starting from the last plane to the first plane in the landing sequence. The last plane is assigned a landing at its latest landing time $L_{N}$ and any of the preceding planes are assigned as late as possible from their target landing time, while maintaining the safety distance constraint. This is ensured by $\min\{PS_{i},L_{i}\}$, where $L_{i}$ is the latest landing time of aircraft $i$ and $PS_{i}$ is the closet landing time possible for aircraft $i$ to all its following aircraft. We define $PS_{i}$ as $PS_{i}= \displaystyle \min_{j=i+1,i+2, \dots, N} (ST_{j} - S_{i,j})$ where any plane $i$ maintains the safety distance with all its following planes. 

If any of the aircraft lands outside its time window with this assignment, then it shows that this landing sequence in infeasible. Since the landings times are assigned as close as possible to their latest landing times, increasing the landing time of any aircraft will cause infeasibility as the last aircraft is landing at its latest landing time and all the preceding planes are scheduled as close as possible. Hence, the optimal solution can be obtained only by decreasing the landing times or leaving them unchanged if there is no reduction possible. 
\end{proof}

Given this initialization, it is possible to reduce the landing times straight away to improve the solution as is depicted in Algorithm~\ref{improve}. Let the initial landing times of the aircraft be assigned according to Equation~\eqref{lemma1} for any given feasible landing sequence. If any aircraft $i$ with $i=1,2, \dots, N$, has a positive deviation $D_{i}$ from its target landing time and maintains a positive extra safety separation $ES_{i}$, then we can decrease the landing time $ST_{i}$ by $\min\{D_{i},ES_{i}\}$. The reason is, this reduction of the landing time is independent of other aircraft as we do not disturb the safety constraint and reduce the landing time of ${i}$ only to bring it closer to its target landing time, which is the requirement of Equation~\eqref{ob}. If $D_{i}>ES_{i}$ then we reduce the landing time by $ES_{i}$ so as to maintain the safety constraint and if $D_{i}<ES_{i}$ then we reduce the landing time to its target landing time. 

However, if the value of $D_{i}\leq 0$ for all the aircraft, then there is no improvement possible and Equation~\eqref{lemma1} is the optimal assignment for this landing sequence with respect to Equation~\eqref{ob}. Note that $ES_{i}\geq 0$ $\forall$ $i$, since the safety distance constraint is always maintained. We present this improvement of the initial landing times for the single runway case in Algorithm~\ref{improve}. We would like to point out that Algorithm~\ref{improve} will not necessarily return the optimal solution but only fetch an improvement to the initial assignment of the landing times.

\begin{algorithm}\caption{Improvement of Individual Landing Times}\label{improve}
\begin{algorithmic}[1]
\STATE Initialization: Equation~\eqref{lemma1}
\STATE $i \leftarrow 1$
\WHILE{$i \neq N+1$}
	\STATE \textbf{Compute} $D_{i}, ES_{i}$
	\IF {$(D_{i} > 0)$}
		\STATE $ST_{i} \leftarrow ST_{i} - \min\{D_{i},ES_{i}\}$
		\STATE \textbf{Update} $D_{i}, ES_{i}$
	\ENDIF
	\STATE $i \leftarrow i+1$
\ENDWHILE
\RETURN $ST,ES,D$
\end{algorithmic}
\end{algorithm}

\begin{lemma}\label{2lemma}
Implementation of Algorithm~\ref{improve} will yield either one of the below mentioned cases for any aircraft ${i}, i=1,2,3,\dots, N$:

\begin{equation}\label{xx}
\begin{split}
\text{a)   }		 D_{i}>0, ES_{i}=0,    &\quad  \text{  b)   }		 D_{i}=0, ES_{i}>0,  \\
\text{c)   }		 D_{i}=0, ES_{i}=0,    &\quad  \text{  d)   }		 D_{i}<0, ES_{i}=0,  \\
\text{e)   }		 D_{i}<0, ES_{i}>0. &\quad 
\end{split}
\end{equation}
\end{lemma}

\begin{proof}
The initialization of the landing times using Equation~\eqref{lemma1} can assign the landing time to any aircraft $i$ anywhere in its time window, if the landing sequence is feasible. Hence, we have the following five cases:\\

\noindent
\textbf{Case 1: $ST_{i} = E_{i}\;.$} \\
If $i=1$, then $D_{i}<0$ and $ES_{i}=0$ from Equation~\eqref{est}. If $i>1$ then $D_{i}<0$ but we need to check for the value of $ES_{i}$. According to Equation~\eqref{est}, we have $ES_{i} \leftarrow ST_{i}-\max\{SP_{i}, E_{i}\}$. Note that $ST_{i}\geq SP_{i},$ $i=2,3,\dots, N$ since the safety separation is always maintained between any two aircraft landing consecutively. This implies that we can write $\max\{SP_{i}, E_{i}\}=ST_{i}$ due to the case constraint, {\em i.e.\ \/}$E_{i}=ST_{i}$. Hence, we have $ES_{i}=0$ from its definition. Since a reduction in the landing time is possible only if $D_{i} > 0$, the values of $D_{i}$ and $ES_{i}$ will remain unchanged by the implementation of Algorithm~\ref{improve}, satisfying Case~\emph{d}.\\

\noindent
\textbf{Case 2: $E_{i} <  ST_{i} < T_{i}\;.$} \\
$D_{i}<0$ for any $i$, which means that the landing time for aircraft ${i}$ will remain unchanged. If $i=1$, then $ES_{i}>0$ from Equation~\eqref{est}. If $i>1$ then again from Equation~\eqref{est} we can deduce that $ES_{i}\geq 0$ because $ST_{i} \geq SP_{i}$ (safety constraint) and $ST_{i}>E_{i}$ (case constraint). This proves that if the initial landing time for any aircraft lies between $E_{i}$ and $T_{i}$ then Algorithm~\ref{improve} will not fetch any reduction hence satisfying Case~\emph{d} or~\emph{e} of Lemma~\ref{2lemma}.\\

\noindent
\textbf{Case 3: $ST_{i} =  T_{i}\;.$} \\
$D_{i}=0$ for any $i$ since the landing occurs at the target landing time. And $ES_{i} \geq 0$ for any $i$, by the same reasons as in Case~\emph{2}. In this case as well there will be no reduction and Case~\emph{b} or~\emph{c} of Lemma~\ref{2lemma} is satisfied.\\

\noindent
\textbf{Case 4: $T_{i} <  ST_{i} < L_{i}\;.$} \\
If the initial landing time for any aircraft ${i}$ lies between $T_{i}$ and $L_{i}$, then $D_{i}>0$ by definition and $ES_{i}\geq 0$ because $ST_{i}>E_{i}$ and $ST_{i} \geq SP_{i}$. Hence, Algorithm~\ref{improve} will reduce the landing time of plane ${i}$ by $\min\{D_{i},ES_{i}\}$. If $\min\{D_{i},ES_{i}\}=D_{i}$ then the reduction in the landing time will fetch $D_{i}=0$ and $ES_{i}>0$, satisfying Case~\emph{b}. If $\min\{D_{i},ES_{i}\}=ES_{i}$ then the reduction in the landing time will fetch $D_{i}>0$ and $ES_{i}=0$, satisfying Case~\emph{a}. However, if after the initialization the values of $D_{i}$ and $ES_{i}$ are equal then the implementation of Algorithm~\ref{improve} will fetch $D_{i}=0$ and $ES_{i}=0$, satisfying Case~\emph{c}. Finally, if $ES_{i}=0$ then there will be effectively no reduction because $\min\{D_{i},ES_{i}\}$ will be equal to zero and Case~\emph{a} of Lemma~\ref{2lemma} will be satisfied.\\

\noindent
\textbf{Case 5: $ST_{i} =  L_{i}\;.$} \\
$D_{i}>0$ and $ES_{i}>0$ after the initialization and yet again the Algorithm~\ref{improve} will fetch either one of Case~\emph{a},~\emph{b} or~\emph{c}, with the same arguments as in Case $4$.
\end{proof}

We now give some additional definitions necessary for the understanding of our main algorithm.
\begin{Definition}\label{pl}
$PL$ is a vector of length $N$ and any element of $PL$ ($PL_{i}$) is the net penalty possessed by any aircraft ${i}$, $i=1,2,\dots, N$. We define $PL_{i}$, $i=1,2,\dots, N$, as
\begin{equation}
PL_{i} =
\begin{cases}
	-g_{i}, & \mbox{if }D_{i} \leq 0 \\
	h_{i}, & \mbox{if }D_{i} > 0\;.
\end{cases}
\end{equation}
\end{Definition}

With the above definition we can now express the objective function stated in Equation~\eqref{ob} in a compact form as
\begin{equation}\label{obn}
\min \sum\limits_{i=1}^{N} (D_{i}\cdot PL_{i})\;.
\end{equation}

\begin{Definition}\label{sigma}
Let $i$ be any aircraft landing at $ST_{i}$ then we define $\sigma(i)$ as the algebraic deviation of the landing time of aircraft $i$ from its earliest landing time $E_{i}$. Mathematically, $\sigma(i)=ST_{i}-E_{i}$, $i=1,2,\dots, N$.
\end{Definition}

\begin{Definition}\label{mu}
Let aircraft $({i},{i+1},\dots, {j})$ be the aircraft in any given sequence which land consecutively in this order on the same runway, we define $\mu$ such that $\mu$ is the last plane in $(i,i+1,\dots,j)$ with $D_{\mu} \leq 0$.
\end{Definition}

\begin{Definition}\label{rsa}
We define $\varGamma=\{({i_{1}}:{j_{1}}),({i_{2}}:{j_{2}}),\dots,  ({i_{c}}:{j_{c}})\}$ as the sets of aircraft which land consecutively, where $c$ is the total number of sets in $\varGamma$ and $1\leq i_{1} < j_{1} < i_{2} < j_{2} < \dots < i_{c} < j_{c}\leq N$. And for any set $\varGamma(k)=({i_{k}}:{j_{k}})$, where the aircraft ${i_{k}},{i_{k}+1},\dots, {j_{k}}$ land one after another consecutively on the same runway, the following properties hold:
\begin{equation}\label{eqrsa}
\begin{cases}
ES_{i_{k}} > 0, \\
ES_{m}=0,   m={i_{k}+1},\dots, {j_{k}} \\
\sum\limits_{\rho={i_{k}}}^{j_{k}} PL_{\rho} >0 \\
\sigma(\rho) >0, \rho={i_{k}},\dots, {j_{k}} \text{ and} \\
\sum\limits_{\rho={\mu}}^{j_{k}} PL_{\rho} >0 \text{ if $\mu$ exists for $\varGamma(k)$\;.}\\
\end{cases}
\end{equation}
\end{Definition}

\begin{Definition}\label{sng}
We define $SNG(X_{i:j})$ as the smallest non-negative number in vector $X$ from elements $X_i$ to $X_j$, $(i<j)$.
\end{Definition}

With the above concepts and definitions we present our main algorithm (Algorithm~\ref{main}) for optimizing a given landing sequence $P$ on a single runway for the special case of the ALP when the safety constraint for any aircraft is to be maintained only with its preceding plane. In other words, when $SP_{i}=ST_{i-1}+S_{i-1,i}$ and $PS_{i}=ST_{i+1}-S_{i,i+1}$. For the general case of the safety constraint the algorithm still returns a feasible solution but not necessarily optimal. Later we show with our results that this special case of the safety constraint holds for several instances and we obtain optimum results for many instances. Moreover we also obtain better results than the best known solutions for several instances.

\begin{algorithm}\caption{Main Algorithm: Single Runway}\label{main}
\begin{algorithmic}[1]
\STATE Apply Algorithm~\ref{improve}
\STATE Calculate $PL,\varGamma,c,\sigma$
\STATE $Sol \leftarrow \sum\limits_{i=1}^{N} (D_{i}\cdot PL_{i})$
\WHILE{$\varGamma \neq \varnothing$}
\FOR {$k = $ 1 \textbf{to} $c$}
\STATE $({i_{k}},{j_{k}}) \leftarrow \varGamma(k)$
\STATE $\gamma=\min\limits_{\rho={i_{k}},\dots, {j_{k}}}\sigma(\rho)$
\STATE $pos = \min(SNG(D_{\varGamma(k)}), ES_{i_k},\gamma)$
\FOR {$p = {i_{k}}$ \textbf{to} ${j_{k}}$}
\STATE $ST_{p}\leftarrow ST_{p}-pos$
\STATE $D_{p}\leftarrow D_{p}-pos$
\ENDFOR
\STATE $ES_{i_{k}}\leftarrow ES_{i_{k}}-pos$
\IF {$j_k<N$}
\STATE $ES_{j_{k}+1}\leftarrow ES_{j_{k}+1}+pos$
\ENDIF
\STATE Update $PL,\sigma$
\ENDFOR
\STATE $Sol \leftarrow \sum\limits_{i=1}^{N} (D_{i}\cdot PL_{i})$
\STATE Calculate $\varGamma,c$
\ENDWHILE
\RETURN $Sol$
\end{algorithmic}
\end{algorithm}

\section{Proof of Optimality}
\begin{theorem}\label{theorem}
Algorithm~\ref{main} returns the optimal value for Equation~\eqref{obn} for any given feasible landing sequence on a single runway when $SP_{i}=ST_{i-1}+S_{i-1,i}$ for $i=2,3,\dots, N$ and $PS_{i}=ST_{i+1}-S_{i,i+1}$ for $i=1,2,\dots, N-1$.
\end{theorem}

\begin{proof}
The initialization of the landing times for any sequence is done according to Lemma~\ref{1lemma}. It allocates the landing times as late as possible, hence the solution can be improved only by reducing the landing times. Thereafter we show that for any aircraft ${i}$ we can reduce its landing time straight away, independent of other aircraft, if $D_{i}>0$ and $ES_{i}>0$. The reason is, if there is an extra safety separation between ${i-1}$ and ${i}$ as well as a positive deviation from the target landing time, then the reduction of $ST_{i}$ by $\min\{D_{i},ES_{i}\}$ will bring aircraft ${i}$ closer to $T_{i}$ and hence yield an overall reduction in the total weighted tardiness thereby improving the overall solution. Note that this reduction will neither cause any aircraft to land earlier than its target landing time nor will it disrupt the safety separation. The implementation of Algorithm~\ref{improve} will fetch one of the five possibilities to all the aircraft, mentioned and proved in Lemma~\ref{2lemma}.

The next step is to prove that a further improvement to the solution is possible iff $\varGamma \neq \varnothing$. If $\varGamma \neq \varnothing$, then we have $ES_{i_k}>0$, $ES_{m}=0,  (m={i_{k}+1},\dots, {j_k})$, $\sum_{\substack{\rho=i_k}}^{j_k} PL_{\rho} >0$, $\sigma(\rho)>0$ where $\rho$ are all the planes in the set $\varGamma(k)$ and $\sum_{\substack{\rho=\mu}}^{j_{k}} PL_{\rho} >0$, if $\mu$ exists. We have $ES_{i_k}>0$ and $ES_{m}=0,  (m={i_{k}+1},\dots, {j_k})$. Reducing the landing time of any aircraft in $m$ will cause infeasibility as it will disrupt the safety constraint since $ES_{m}=0$. But reducing the landing times of all the aircraft from ${i_k}$ to ${j_k}$ by $pos=\min(SNG(D_{\varGamma(k)}), ES_{i_k},\gamma)$ will not cause any infeasibility for two reasons. First, the definition of $\varGamma(k)$ ensures that all the planes have a positive deviation from their earliest landing times since $\sigma(\rho)>0$ and the reduction of the landing times by $pos$ will not cause any infeasibility since all the aircraft in set $\varGamma(k)$ will be allocated a landing time within their time window since $pos\leq \gamma$. Second, we would reduce all the landing times by the same amount and not more than $ES_{i_k}$. This will maintain the safety separation between all the aircraft in $\varGamma(k)$ and also the required separation between aircraft ${i_{k}-1}$ and ${i_k}$. 

Notice that $PL_{\rho}$ is the net penalty of aircraft $\rho$ as stated in Definition~\ref{pl}. Hence a positive value for the summation of the net penalties of aircraft ${i_{k}}$ to ${j_k}$ landing consecutively means, that the total tardiness penalty is higher than the total earliness penalty and an increase in the landing times of all the aircraft in $\varGamma(k)$ by the same amount is only going to worsen the solution. As for $\mu$, let's say there exists a $\mu$ for the set $\varGamma(k)$ such that $\sum_{\substack{\rho=\mu}}^{j_{k}} PL_{\rho} < 0$. This shows that aircraft ${\mu}$ to ${j_{k}}$ already possess a net earliness penalty and further reducing their landing times will fetch an increase in the overall penalty. However, $\sum_{\substack{\rho=i_k}}^{j_k} PL_{\rho} >0$ means that $\sum_{\substack{\rho=i_k}}^{\mu -1} PL_{\rho} >0$ which implies that aircraft ${i_{k}}$ to ${\mu-1}$ possess a net positive tardiness penalty. Thus, a reduction in landing times by $\min(SNG(D_{i_{k}:{\mu-1}}),ES_{i_k},\gamma)$ will reduce the total weighted tardiness as well as ensure that the increase in the earliness penalty (if any) of aircraft ${i_{k}}$ to ${\mu-1}$ does not exceed the reduction in the net tardiness penalty and thereby reducing the overall penalty. In such a case $\varGamma(k)$ will become $({i_{k}}:{\mu-1})$.

Conversely, if $\varGamma = \varnothing$, then either one of the cases will not hold in Definition~\ref{rsa}. We prove this by contradiction for all these cases:\\

\noindent
\textbf{Case 1:} $ES_{i_{k}} > 0 \;.$ \\
If $ES_{i_{k}} = 0 $ and all the other conditions hold then there is no scope of reduction and an increase in the landing times will only worsen the solution. Note that $ES_{i_{k}}$ will never be negative, for any $i_{k}$, $i_{k}=1,2,\dots, N-1$.\\

\noindent
\textbf{Case 2:} $ES_{m}=0,m={i_{k}+1},\dots, {j_{k}}\;.$ \\
If $ES_{m} \neq 0,(m={i_{k}+1},\dots, {j_{k}})$
then we have two cases. One, if for some $m$, $ES_{m}<0$, then the solution is infeasible. Second, if for some $m$, $ES_{m}>0$ then it contradicts the definition of $\varGamma$.\\

\noindent 
\textbf{Case 3:}
$\sigma(\rho) >0, \rho={i_{k}},\dots, {j_{k}}\;.$ \\
If the value of $\sigma(\rho) = 0$, then a reduction of the landing times for all the planes in the set $\varGamma(k)$ by any positive value will make the solution infeasible since the aircraft ${\rho}$ is already landing at its earliest landing time. Note that the value of $\sigma(\rho)$ cannot be negative for any aircraft ${\rho}$ at any stage.\\

\noindent 
\textbf{Case 4:}
$\sum_{\substack{\rho=i_{k}}}^{j_{k}} PL_{\rho} >0\;.$ \\
If $\sum_{\substack{\rho=i_{k}}}^{j_{k}} PL_{\rho}  = 0$ for any plane ${\rho}$, then any change to the landing times of all the aircraft in $\varGamma(k)$ will only worsen the solution by increasing the total lateness penalty or the total earliness penalty. If $\sum_{\substack{\rho=i_{k}}}^{j_{k}} PL_{\rho}  < 0$, then the reduction of landing times is again going to worsen the solution as the total earliness penalty is already higher than the total lateness penalty. Moreover, an increase in the landing time is not good either, because it will only take us back to an earlier step where $\sum_{\substack{\rho=i_{k}}}^{j_{k}} PL_{\rho}  > 0$.\\

\begin{lemma}\label{pos}
If $\varGamma(k) \neq \varnothing$ then $pos$ exists and $pos>0$.
\end{lemma}

\begin{proof}
From Algorithm~\ref{main} we have, $pos=\min(SNG(D_{\varGamma(k)}), ES_{i_k},\gamma)$. So $pos$ will exist with a positive value only if $SNG(D_{\varGamma(k)})>0$, $ES_{i_{k}}>0$ and $\gamma>0$. Clearly, $ES_{i_{k}}>0$ from the definition of $\varGamma(k)$. Besides, $ES_m = 0$ for $m={i_{k}+1},\dots, {j_{k}}$ and $\sum_{\substack{m=i_{k}}}^{j_{k}} PL_{m} >0$ again from Equation~\eqref{rsa}. Note that we proved in Lemma~\ref{2lemma} that if $ES_i=0$ for any $i=1,2,3,\dots, N$ then $D_{i}\leq 0$. Moreover, $\sum_{\substack{m=i_{k}}}^{j_{k}} PL_{m} >0$ shows that for at least one aircraft $m$ in the $\varGamma(k)$ has $PL_{m}>0$. Recall from Equation~\eqref{pl} that for any aircraft $m$, $PL_{m}>0$ only if $D_{m}>0$. Thus we have $ES_{i_{k}}>0$ and $D_{m}>0$ at least for one aircraft $m$, where $m={i_{k}},{i_{k}+1},\dots, {j_{k}}$. Furthermore, if $\varGamma \neq \varnothing$ then obviously $\gamma>0$ since the $\gamma=\min\limits_{\rho={i_{k}},\dots, {j_{k}}}\sigma(\rho)$ and $\sigma(\rho)>0$ for all the aircraft in the set $\varGamma(k)$ from Equation~\eqref{eqrsa}. Hence, this proves that $pos$ will exist and will be greater than zero if $\varGamma(k)\neq \varnothing$.
\end{proof}

We reduce the landing times by $\min(SNG(D_{\varGamma(k)}), ES_{i_k},\gamma)$ because this will neither disrupt the safety constraint nor cause infeasibility. Besides, this will not alter the number of planes arriving early $(D_{m}<0)$. If we reduce the landing times by a greater quantity we will certainly reduce the lateness penalty but we might as well end up increasing the earliness penalty by a greater amount. Hence we do not want to change the number of planes arriving early. Notice that a reduction in the landing time for aircraft ${j_k}$ by $pos$ means that it will increase the extra safety separation between ${j_k}$ and ${j_{k}+1}$, which is why we have line $14$ in Algorithm~\ref{main}. Hence, to summarize, Algorithm~\ref{main} allocates the latest possible initial landing times to all the aircraft and then makes improvements to the solution at every step until there is no improvement possible.
\end{proof}

\section{Multiple Runways}
In this section we propose an effective approach for allocating the runways to all the aircraft in a given landing sequence. We do not prove the optimality of this approach but our results show that it is an effective strategy and performs better than other approaches mentioned in the literature. In the multiple runway case the only difference is the initial assignment of the runways to all the aircraft in a given sequence. We propose an initialization algorithm for the multiple runway case which again takes the input as a landing sequence of planes waiting to land and the number of runways $R$ at the airport. We make an assumption as in~\cite{pinol}, that if aircraft ${i}$ and $j$ land on different runways then $S_{i,j}=0$. Proposition~\eqref{prop} assigns the appropriate runway to all the aircraft and the landing sequence on each runway. Let ${A_{1r}, A_{2r},\dots, A_{nr}}$ be the sequence of planes on runway $r, r = 1,2,\dots, R$ and $nr$ be the number of planes landing on runway $r$.

\begin{proposition}\label{prop}
Assign the first $R$ air planes ${1},{2},\dots, {R}$, one on each runway at their respective target landing times. For any following aircraft ${i}, i=R+1,R+2,\dots, N$ assign the same runway as ${i-1}$ at a landing time of $T_{i}$ if $T_{i}$ is greater than or equal to the allowed landing time for plane $i$ by maintaining the safety distance constraint with all the preceding aircraft on the same runway. Otherwise, assign a runway $r$ at $T_{i}$ which offers zero deviation from $T_{i}$. If none of the above two conditions hold then select a runway which gives the least feasible positive deviation to plane ${i}$ from its target landing time.
\end{proposition}

Here we make an obvious assumption that the number of air planes waiting to land is more than the number of runways present at the airport. The landing sequence in this proposition is maintained in the sense that any aircraft ${i}$ does not land before ${i-1}$ lands. Once we have this assignment of aircraft to runways, each runway has a fixed number of planes landing in a known sequence. Using this to our benefit, we can now apply Algorithm~\ref{main} to each runway separately. We would mention here that in this work we assume the safety separation time between aircraft landing on different runways to be equal to zero, {\em i.e.\ \/}$s_{A_{ir},A_{jr^{\prime}}}=0$, where $r$ and $r^{\prime}$ are two different runways. This assumption was also considered by Pinol~\emph{et al.} in~\cite{pinol}.

\section{Algorithm Run-Time Complexity}
\begin{lemma}
The run-time complexity of Algorithm~\ref{main} is $O(N^3)$ where $N$ is the total number of aircraft.
\end{lemma}
\begin{proof}
Calculation of $\varGamma$ requires finding all the sets of planes landing consecutively, such that they hold certain properties as mentioned in Equation~\eqref{eqrsa}. The worst case scenario for the calculation of $\varGamma$ will occur when every aircraft lies in one of the sets of $\varGamma$. Let any set $\varGamma(k)$ has $x_{k}$ aircraft where $k=1,2,\dots,c$, then we have, $\sum_{\substack{k=1}}^{c} x_{k}=N$, since all the sets of $\varGamma$ are disjoint. The run-time for the calculating and checking the first four properties of any set $\varGamma(k)$ is $O(x_{k})$. However the computation of $\mu$ and checking $\sum_{\substack{\rho=\mu}}^{j_{k}} PL_{\rho} >0$ requires a computation of all the prior properties, if $\mu$ exists. In the worst case the value of $j_{k}$ will drop down to $i_{k}+1$ and this would require a total run-time of $O(x_{k}^2)$ where $x_{k}$ is the number of aircraft in the set $\varGamma(k)$ obtained initially by the computation of the first four properties in Equation~\eqref{eqrsa}. Let $T$ be the run-time of the computation of all the sets of $\varGamma$. Since all the properties are calculated in a sequential manner, we have, $T= \sum_{\substack{k=1}}^{c} O({x_{k}}^2)$. Besides, $x_{k}>0$ for $k=1,2,\dots,c$, we can write $\sum_{\substack{k=1}}^{c} O({x_{k}}^2) \leq \left(\sum_{\substack{k=1}}^{c} O(x_{k})\right)^2 $. Now using $\sum_{\substack{k=1}}^{c} x_{k}=N$ we get $T=O(N^2)$. It is straight forward to observe that the complexity of Algorithm~\ref{improve} is $O(N^2)$ due to the calculations of $ST$ and $ES$. The computation of $PL$ and $Sol$ in Algorithm~\ref{main} are both of $O(N)$ each. The $while$ loop in line $4$ involves several iterations so we first study the run-time of a single iteration of the $while$ loop. The $for$ loop in line $5$ is run for the number of sets in $\varGamma$. Hence, the total run-time of the $for$ loop is $\sum_{\substack{k=1}}^{c} O(x_{k})$, which is again equal to $O(N)$. The next steps inside the $while$ loop involve the computation of $Sol$ with a run-time of $O(N)$ and all the sets of $\varGamma$ which requires $O(N^2)$ run-time each at every iteration. Since the computation of $Sol$ and $\varGamma$ is carried out sequentially, the total run-time complexity of the algorithm is basically equal to $O(\lambda N^2)$, where $\lambda$ is the number of times the $while$ loop is iterated. Clearly, the maximum value of $\lambda$ can be equal to the maximum number of aircraft in any set $\varGamma(k)$, which is equal to the total number of aircraft $N$. Hence the run-time complexity of Algorithm~\ref{main} is $O(N^3)$.
\end{proof}
\section{Results}
We now present our results for the aircraft landing problem with single and multiple runways for the benchmark instances provided by Beasley~\cite{beasley1}. We implement the algorithm as described above to find the optimal solution for the special case of the ALP in conjunction with Simulated Annealing (SA). We use a slightly modified Simulated Annealing algorithm to generate the landing sequences and Algorithm~\ref{main} to optimize each sequence to its minimum penalty. The ensemble size for SA is taken to be $20$ for all the instances. The initial temperature is kept as twice the standard deviation of the energy at infinite temperature: $\sigma_{E_{T=\infty}} = \sqrt{\langle E^2 \rangle_{T=\infty} - \langle E \rangle^2_{T=\infty}}$. We estimate this quantity by randomly sampling the configuration space~\cite{salamon}. An exponential schedule for cooling is adopted with a cooling rate of $0.999$. One of the modifications from the standard SA is in the acceptance criterion. We implement two acceptance criteria: the Metropolis acceptance probability, $\min\{1,\exp((-\hspace{-0.3em}\bigtriangleup\hspace{-0.3em} E)/T)\}$~\cite{salamon} and a constant acceptance probability of $0.07$. A solution is accepted with this constant probability if it is rejected by the Metropolis criterion. This concept of a constant probability is useful when the SA is run for many iterations and the metropolis acceptance probability is almost zero, since the temperature would become infinitesimally small. 

Apart from this, we also incorporate elitism in our modified SA. Elitism has been successfully adopted in evolutionary algorithms for several complex optimization problems~\cite{elitist1,elitist2}. As for the perturbation rule, we first randomly select a certain number of aircraft in any given landing sequence and permute them randomly to create a new sequence. The number of planes selected for this permutation is taken as  $3 + \lfloor\sqrt{N/50}\rfloor$, where $N$ is the number of aircraft. For large instances the size of this permutation is quite small but we have observed that it works well with our modified simulated annealing algorithm. We take the initial landing sequence for our algorithm as the sequence as per their target landing times.

\begin{table}
\centering
{\scriptsize
\caption{Results for small benchmark instances and comparison with Scatter Search and the Bionomic Algorithm~\cite{pinol}.}
\label{result1}
\begin{tabular}{|p{0.3cm} p{0.3cm} p{0.95cm}| p{1.0cm} p{0.6cm} p{0.7cm}| p{1.0cm} p{0.6cm} p{0.95cm} | p{1cm} p{0.7cm} p{0.7cm}|} \hline
\textbf{N}&\textbf{R}&$\mathbf{Z_{\textrm{opt}}}$&\multicolumn{3}{c|}{\textbf{SCS}}  & \multicolumn{3}{c|}{\textbf{BA}}  & \multicolumn{3}{c|}{\textbf{PSA}} \\ \cline{4-6} \cline{7-9} \cline{10-12}
&         &           & $\mathbf{Z_{\textrm{SCS}}}$ &$\mathbf{T_{\textrm{run}}}$& $\mathbf{G_{\textrm{best}}}$ & $\mathbf{Z_{\textrm{BA}}}$ &$\mathbf{T_{\textrm{run}}}$& $\mathbf{G_{\textrm{best}}}$ & $\mathbf{Z_{\textrm{PSA}}}$ &$\mathbf{T_{\textrm{run}}}$& $\mathbf{G_{\textrm{best}}}$ \\ \hline
\multirow{3}{*}{10} & 1 &  700  &     700 &  0.4 &     0 &     700 &    6 &     0 &   700 & 0.006 & 0 \\
					& 2 &   90  &      90 &  2.4 &     0 &      90 &  4.5 &     0 &    90 & 0.211 & 0 \\
					& 3 &    0  &       0 &  3.9 &     0 &       0 &  3.4 &     0 &     0 & 0.243 & 0 \\ \hline
\multirow{3}{*}{15} & 1 &  1480 &    1480 &  0.6 &     0 &    1480 &    9 &     0 &  1480 & 0.095 & 0 \\
					& 2 &   210 &     210 &  4.5 &     0 &     210 &  4.9 &     0 &   210 & 0.312 & 0 \\
					& 3 &     0 &       0 &  4.6 &     0 &       0 &  4.3 &     0 &     0 & 0.290 & 0 \\ \hline
\multirow{3}{*}{20} & 1 &  820  &     820 &  0.8 &     0 &     820 &  9.9 &     0 &   820 & 0.300 & 0 \\
					& 2 &   60  &      60 &  4.8 &     0 &      60 &  5.8 &     0 &    60 & 0.363 & 0 \\
					& 3 &    0  &       0 &  6.2 &     0 &       0 &  6.3 &     0 &     0 & 0.381 & 0 \\ \hline
\multirow{4}{*}{20} & 1 &  2520 &    2520 &  0.8 &     0 &    2520 &  9.5 &     0 &  2520 & 0.014 & 0 \\
					& 2 &   640 &     640 &  5.2 &     0 &     640 &  5.5 &     0 &   640 & 0.352 & 0 \\
					& 3 &   130 &     130 &  4.6 &     0 &     130 &  5.7 &     0 &   130 & 0.371 & 0 \\
					& 4 &     0 &       0 &  5.6 &     0 &       0 &  5.2 &     0 &     0 & 0.377 & 0 \\ \hline
\multirow{4}{*}{20} & 1 &  3100 &    3100 &  0.9 &     0 &    3100 &   10 &     0 &  3100 & 0.285 & 0 \\
					& 2 &   650 &     650 &    5 &     0 &  670.02 &  6.1 &  3.08 &   650 & 2.230 & 0 \\
					& 3 &   170 &     170 &  5.4 &     0 &     170 &  4.3 &     0 &   170 & 0.456 & 0 \\
					& 4 &     0 &       0 &  5.6 &     0 &       0 &  6.8 &     0 &     0 & 0.507 & 0 \\ \hline
\multirow{3}{*}{30} & 1 & 24442 &   24442 & 15.8 &     0 &   24442 & 27.4 &     0 & 24442 & 0.002 & 0 \\
					& 2 &   554 &     554 &    7 &     0 &  573.99 & 10.1 &  3.61 &   554 & 2.629 & 0 \\
					& 3 &     0 &       0 &  5.4 &     0 &       0 &  8.7 &     0 &     0 & 0.297 & 0 \\ \hline
\multirow{2}{*}{44} & 1 &  1550 &    1550 & 19.5 &     0 &    1550 &  7.9 &     0 &  1550 & 0.015 & 0 \\
					& 2 &     0 &       0 & 11.8 &     0 &       0 & 12.4 &     0 &     0 & 0.345 & 0 \\ \hline
\multirow{3}{*}{50} & 1 &  1950 & 2964.97 &  4.2 & 52.05 & 2654.92 & 28.7 & 36.15 &  1995 & 3.915 & 2.31 \\
					& 2 &   135 &     135 & 12.1 &     0 &     135 & 19.6 &     0 &   135 & 4.357 & 0 \\
					& 3 &     0 &       0 & 13.9 &     0 &       0 & 18.1 &     0 &     0 & 0.421 & 0 \\ \hline
\multicolumn{3}{|c|}{\textbf{{\scriptsize Average}}}&  & 6.0  &   2.1 &         &  9.6 &   1.7 &       & 0.75 & 0.092 \\ \hline
\end{tabular}
}
\end{table}
\vspace{0.5em}
All the computations were carried out in MATLAB on a $1.73$ GHz machine with $2$ GB RAM. To better explain and compare our results we first define some new parameters used in Table~\ref{result1} and~\ref{result2}. Most of these parameters are derived from Pinol~\emph{et al.}~\cite{pinol} with slight changes as explained below.

\noindent
Let, \\
$Z_{\textrm{opt}}$ = the value of the optimal solution, \\
$Z_{\textrm{best}}$ = the best known solutions for ALP provided in~\cite{pinol}, \\
$Z_{\textrm{SCS}}$ = the best solutions obtained in~\cite{pinol} using Scatter Search (SCS),\\
$Z_{\textrm{BA}}$ = the best solutions obtained in~\cite{pinol} using the Bionomic Algorithm (BA), \\
$Z_{\textrm{PSA}}$ = the best solutions obtained in this work, \\
$T_{\textrm{run}}$ = the average run-time in seconds over $10$ replications, \\
$G_{\textrm{best}}$ = percentage gap between the best obtained results and $Z_{\textrm{opt}}$ if the optimal solution known and $Z_{\textrm{best}}$ if the optimal solution is not known.

\begin{table}
\centering
{\scriptsize
\caption{Results for large benchmark instances and comparison with the Scatter Search and the Bionomic Algorithm~\cite{pinol}.}
\label{result2}
\begin{tabular}{|p{0.3cm} p{0.3cm} p{0.95cm}| p{1.0cm} p{0.6cm} p{0.7cm}| p{1.0cm} p{0.6cm} p{0.95cm} | p{1cm} p{0.7cm} p{0.7cm}|} 
\hline
\textbf{N}& \textbf{R} & $\mathbf{Z_{\textrm{best}}}$ & \multicolumn{3}{c|}{\textbf{SCS}}  & \multicolumn{3}{c|}{\textbf{BA}}  & \multicolumn{3}{c|}{\textbf{PSA}} \\ \cline{4-6} \cline{7-9} \cline{10-12}
&         &           & $\mathbf{Z_{\textrm{SCS}}}$ &$\mathbf{T_{\textrm{run}}}$& $\mathbf{G_{\textrm{best}}}$ & $\mathbf{Z_{\textrm{BA}}}$ &$\mathbf{T_{\textrm{run}}}$& $\mathbf{G_{\textrm{best}}}$ & $\mathbf{Z_{\textrm{PSA}}}$ &$\mathbf{T_{\textrm{run}}}$& $\mathbf{G_{\textrm{best}}}$ \\ \hline
\multirow{4}{*}{100}& 1 & 5611.70 & 7298.57 &11.9& 30.06 & 6425.95 &55.4& 14.51 & 5703.54 &14.294& 1.637 \\
					& 2 &  452.92 &  478.6 &34.2&  5.67 & 700.80 &48.7& 54.73 & \textbf{444.1}  &10.78& \textbf{*} \\
					& 3 &   75.75 &  75.75 & 39&     0 & 142.00 &46.6& 87.46 & 75.75  &0.868& 0 \\
					& 4 &       0 &      0 &33.6&     0 & NA &43.9&     n/d & 0  &0.027& 0 \\ \hline
\multirow{5}{*}{150}& 1 & 12329.31 &  17872.56 &22.7& 44.96 & 16508.94 &92.5&  33.90 & 13515.68 & 31.411& 9.62 \\
					& 2 &  1288.73 & 1390.15 &60.8&  7.87 & 1623.15 &84.5&  25.95 & \textbf{1203.76}  &29.090& \textbf{*} \\
					& 3 &   220.79 &  240.39 &66.8&  8.88 & 653.27 &80.3& 195.88 &  \textbf{205.21}  &19.010& \textbf{*} \\
					& 4 &    34.22 &  39.94 &64.7& 16.74 & 134.27 &78.8& 292.40 & 34.22 &3.532& 0 \\
					& 5 &        0 &      0 &60.7&     0 & NA &76.2&      n/d &  0  &0.0171& 0 \\ \hline
\multirow{5}{*}{200}& 1 & 12418.32 &  14647.40 &25.6& 17.95 & 14488.45 & 141.7 & 16.67 & 13401.57 &27.782& 7.92 \\
					& 2 &  1540.84 & 1682.44 &95.9&  9.19 & 2134.67 &128.7&  38.54 & \textbf{1400.64}  & 43.77 & \textbf{*} \\
					& 3 &   280.82 & 341.44 &102.1& 21.59 & 1095.45 &120.3& 290.09 &  \textbf{253.15}  &11.125& \textbf{*} \\
					& 4 &    54.53 &  56.04 &99.3&  2.77 & 313.25 &116.8& 474.47 & 54.53  &0.0245& 0 \\
					& 5 &        0 &      0 &95.6&     0 & NA &115.8&      n/d &  0  &0.0230& 0 \\ \hline
\multirow{5}{*}{250}& 1 & 16209.78 & 19800.24 &38.1&  22.15 & 20032.04 &201.1 &    23.58 & 17346.45
 &34.93& 7.01 \\
					& 2 &  1961.39 &  2330.13 &126.6&  18.80 & 2945.61 &183.5&    50.18 & \textbf{1753.67}  &47.24& \textbf{*} \\
					& 3 &   290.04 &   340.73 &145.4&  17.48 & 864.34 &171&   198.01 &  \textbf{233.49}  &16.271& \textbf{*} \\
					& 4 &     3.49 &    12.96 &144.5& 271.63 & 464.76 &168.8& 13216.91 & \textbf{2.44}  &1.324& \textbf{*} \\
					& 5 &        0 &        0 &138.6&      0 & NA &166.2&   n/d & 0  &0.0308& 0 \\ \hline
\multirow{5}{*}{500}& 1 & 44832.28 & 46284.84 &123.7&  3.24 & 45294.15 &585.2&     1.03 & \textbf{43052.04} &52.717& \textbf{*} \\
					& 2 &  5501.96 &  5706.63 &383.6&  3.72 & 7563.54 &537.9&    37.47 & \textbf{4593.77}  &48.223& \textbf{*} \\
					& 3 &  1108.51 &  1130.45 &456&  1.98 & 3133.64 &515.8&   182.69 & \textbf{712.81}  &45.168& \textbf{*} \\
					& 4 &   188.46 &   231.76 &441.3& 22.98 & 2425.12 &497.7&  1186.81 & \textbf{89.95}  &48.6& \textbf{*} \\
					& 5 &     7.35 &     7.35 &442.1&     0 & 1647.02 &488.7& 22308.44 &  \textbf{0}  &0.0554& \textbf{*} \\ \hline
\multicolumn{3}{|c|}{\textbf{Average}}	&     &135.5  &   22.0 &   & 197.8 & 1936.5 &  &  20.263 & 1.091 \\ \hline
\multicolumn{12}{l}{NA: Results not available.}
\end{tabular}
}
\end{table}

$G_{\textrm{best}}$ is defined as $G_{\textrm{best}}$ $=$ $100$ $\cdot$ ($Z_{\textrm{PSA}}-Z_{\textrm{best}})/Z_{\textrm{best}}$; if the optimal solution is known then $Z_{\textrm{best}}=Z_{\textrm{opt}}$. However, if $Z_{\textrm{best}}=0$ we follow the same notation as assumed in~\cite{pinol}. If $Z_{\textrm{best}}=0$, then the value of $G_{\textrm{best}}=0$ if the best solution obtained is also zero and n/d (not defined) if the best solution obtained is greater than zero. This definition of $G_{\textrm{best}}$ is the same as explained by Pinol~\emph{et al.}~\cite{pinol}. If for any instance the result obtained by us is better than the best known solution then we denote $G_{\textrm{best}}$ by an asterisk (*). The values of $Z_{\textrm{best}}$ are the best results obtained by Pinol~\emph{et al.}~\cite{pinol} during the course of their work. The results shown in Table~\ref{result1} and~\ref{result2} are obtained by using Algorithm~\ref{main}, Proposition~\ref{prop} and simulated annealing depending on single or multiple runways. For the single runway case we use simulated annealing to generate the landing sequences and Algorithm~\ref{main} to optimize each sequence. For the multiple runway case we first generate a complete landing sequence of all the aircraft using simulated annealing, allocate the aircraft and their landing sequence to each runway using Proposition~\ref{prop} and then apply Algorithm~\ref{main} to each runway separately for optimization. For brevity we call this approach $PSA$. It is clear from Table~\ref{result1} that our approach is much faster and finds the optimal solution for all benchmark instances except for one. The reason that the optimum is found for all other instances is that the optimal sequences for all the remaining instances hold the special case of the safety constraint,~\emph{i.e.,\ \/}the safety constraint for any aircraft depends only on its preceding plane. However for the instance '$airland8$' with $50$ aircraft and a single runway, our algorithm does not return the optimal solution as the optimal landing sequence does not satisfy the special case of the safety constraint.

Nevertheless, our approach still achieves a better result than Scatter Search and the Bionomic Algorithm, with a percentage gap of just $2.31$ percent from the optimal value in less than $4$ seconds. The average percentage gap for our approach on all the benchmark instances is $0.09$ percent as opposed to $2.1$ percent and $1.7$ percent for Scatter Search and the Bionomic Algorithm, respectively. Moreover the average run-time for $PSA$ is just $0.75$ seconds, which is $8$ times faster than Scatter Search with $6.0$ seconds and more than $12$ times faster than the Bionomic Algorithm with an average run-time of $9.6$ seconds. Note that Pinol~\emph{et al.}~\cite{pinol} implemented their algorithms using C++ on a $2$ GHz Pentium PC with $512$ MB memory. 

Table~\ref{result2} presents our results for the large instances. The optimal solutions of these instances are unknown and hence we compare our results with the best known solutions. Not only do we obtain better results than the previous approaches, we also achieve better results than the best known values for $13$ out of $24$ instances. We are unable to reach the best known solutions for four instances but in general we perform much better than Scatter Search and the Bionomic Algorithm. The maximum percentage gap for any of these instances with the best known solutions is $9.62$ percent as opposed to a percentage gap of $44.96$ percent with Scatter Search and $33.90$ percent with the Bionomic Algorithm, for the same set of instances. Again, the average percentage gap for our approach is $1.091$ percent as opposed to $22.0$ percent and $1936.5$ percent for Scatter Search and the Bionomic Algorithm, respectively. Moreover, the average run-time for $PSA$ is just $20.263$ seconds which is again much faster than Scatter Search with $135.5$ seconds and the Bionomic Algorithm with an average run-time of $197.8$ seconds, for all the instances in Table~\ref{result2}. Hence, we show that the use of our polynomial algorithm fetches faster and better results than previous approaches. We would like to mention here that although we do not prove that Proposition~\eqref{prop} returns optimal results, nevertheless we obtain optimal solutions for all the small instances in much less time. For the large instances, the results again show that it is an effective approach and yields better and faster results for all the instances.

\section{Conclusion}
The Aircraft landing problem has mostly been approached using linear programming, meta-heuristic approaches or branch and bound algorithms in the last two decades~\cite{krishna,earnst,beasley,salehipour}. This paper is the first attempt to schedule the landings of the aircraft for a given feasible landing sequence using a polynomially bound algorithm. We have tested our approach over all the benchmark instances which have been applied in major previous research and our results show that we are able to find better solutions than the best known solutions for many instances. In future we intend to optimize our algorithm for the general case of the safety constraint for all the benchmark instances accordingly. The authors are willing to provide the extended results for the results obtained in this work for any (or all) instance(s) via email.

\section*{Acknowledgement}
The research project was promoted and funded by the European Union and the Free State of Saxony. The authors take the responsibility for the content of this publication.

\end{document}